\documentclass[submission,copyright,creativecommons]{eptcs}
\providecommand{\event}{SOS 2007} 
\usepackage{breakurl}             
\usepackage{underscore}           
\usepackage{enumitem}
\usepackage{amsthm}
\usepackage{amssymb}
\usepackage{amsmath}
\usepackage{graphicx}
\usepackage{relsize}
\usepackage{thmtools}
\usepackage{comment}
\usepackage{todonotes}
\usepackage{multicol}
\usepackage{titlesec}

\titlespacing\section{0pt}{4pt plus 6pt minus 2pt}{4pt plus 6pt minus 2pt}
\titlespacing\subsection{0pt}{4pt plus 6pt minus 2pt}{4pt plus 6pt minus 2pt}
\titlespacing\subsubsection{0pt}{4pt plus 6pt minus 2pt}{4pt plus 6pt minus 2pt}

\declaretheorem[style=definition,qed=$\triangle$,numberwithin=section]{definition}

\declaretheorem[style=theorem, sibling=definition]{theorem}
\declaretheorem[style=theorem, sibling=definition]{lemma}
\declaretheorem[style=theorem, sibling=definition]{proposition}
\declaretheorem[style=theorem, sibling=definition]{corollary}

\newcommand{\ld}{\rotatebox[origin=c]{45}{\footnotesize	$\square$}}
\newcommand{\lb}{\rotatebox[origin=c]{0}{$\footnotesize \square$}}

\title{Filtration and canonical completeness \\ for continuous modal $\mu$-calculi}
\author{Jan Rooduijn\footnote{The research of this author has been made possible by a grant from the Dutch Research Council NWO, project nr. 617.001.857.}
	\institute{ILLC\\
		University of Amsterdam\\}
	\email{j.m.w.rooduijn@uva.nl}
	\and
Yde Venema
	\institute{ILLC\\
	University of Amsterdam\\}
\email{y.venema@uva.nl}
}
\def\titlerunning{Filtration and canonical completeness for $\mu_c \mathsf{ML}$}

\def\authorrunning{J.M.W. Rooduijn and Y. Venema}
\begin{document}
	\maketitle
	\begin{abstract}
	The continuous modal $\mu$-calculus is a fragment of the modal $\mu$-calculus, where the application of fixpoint operators is restricted to formulas whose functional interpretation is Scott-continuous, rather than merely monotone. By game-theoretic means, we show that this relatively expressive fragment still allows two important techniques of basic modal logic, which notoriously fail for the full modal $\mu$-calculus: filtration and canonical models. In particular, we show that the Filtration Theorem holds for formulas in the language of the continuous modal $\mu$-calculus. As a consequence we obtain the finite model property over a wide range of model classes. Moreover, we show that if a basic modal logic $\mathsf{L}$ is canonical and the class of $\mathsf{L}$-frames admits filtration, then the logic obtained by adding continuous fixpoint operators to $\mathsf{L}$ is sound and complete with respect to the class of $\mathsf{L}$-frames. This generalises recent results on a strictly weaker fragment of the modal $\mu$-calculus, \textit{viz.} $\mathsf{PDL}$. 
	\end{abstract}
	
	\section{Introduction}
	\paragraph{Filtration and canonical models} This paper concerns two key methods in the theory of modal logic, both of which were introduced in their modern forms by Lemmon \& Scott in~\cite{lemmon1977introduction}. First, \emph{filtration}, which allows one to shrink a Kripke model into a finite one, by identifying states that agree on the truth of some given finite set of formulas. The \textit{Filtration Theorem} then states that the equivalence classes in the finite model satisfy the same formulas as their members do in the original model. Filtration is the most important tool for proving the finite model property and the decidability of modal logics. For an overview of recent developments in the theory of filtration, see~\cite{vanmodern}. 
	
	The other method central to this paper is that of \textit{canonical models}. This well-known technique for proving the completeness of modal logics is related to Henkin's method for first-order logic. Given a modal logic $\mathsf{L}$, it allows one to construct the \textit{canonical model} $\mathbb{S}^\mathsf{L}$ of $\mathsf{L}$ with the powerful property that a formula $\varphi$ is consistent in the logic $\mathsf{L}$ if and only if it is satisfiable in $\mathbb{S}^\mathsf{L}$. It follows that $\mathsf{L}$ is complete with respect to any class of frames containing the \textit{canonical frame}, \textit{i.e.} the frame underlying $\mathbb{S}^\mathsf{L}$. Thus, when then the canonical frame is a frame for $\mathsf{L}$ - in this case $\mathsf{L}$ is said to be \textit{canonical} - the logic $\mathsf{L}$ is complete with respect to the class of frames for $\mathsf{L}$.
	
	\paragraph{Modal fixpoint logics}
	Modal fixpoints logics are extensions of basic modal logic by operators capable of expressing certain kinds of recursive statements. They are of particular interest for computer science, where they are used to express important properties of processes. Examples of modal fixpoint logics are common knowledge logic ($\mathsf{CKL}$), provability logic ($\mathsf{GL}$), propositional dynamic logic ($\mathsf{PDL}$) and computation tree logic ($\mathsf{CTL}$). The central modal fixpoint logic, in which each of the aforementioned logics can be interpreted, is the modal $\mu$-calculus ($\mu \mathsf{ML}$), introduced by Kozen in~\cite{kozen1983results}. It extends basic modal logic with \textit{explicit} least and greatest fixed point operators, resulting in a large gain of expressive power. Although many desirable properties, such as decidability and bisimulation invariance, withstand this gain in expressive power, the methods of filtration and canonical models do not. 
	
	In fact, the method of canonical models breaks down already in the case of relatively simple modal fixpoint logics. The reason is that these logics generally lack the compactness property, preventing the use of infinite maximally consistent sets. If, however, the method of filtration \textit{does} work for such a logic $\mathsf{L}$, then the canonical model method can often be salvaged. This roughly works as follows. One begins by taking the canonical model $\mathbb{S}^\mathsf{L}$. Due to the compactness failure, this model is \textit{non-standard}, meaning that the frame underlying $\mathbb{S}^\mathsf{L}$ fails to satisfy some desired properties. However, by applying filtration to $\mathbb{S}^\mathsf{L}$ we obtain a finite model (a \textit{finitary canonical model}), whose underlying frame often does satisfy these desired properties. This procedure for instance underlies the completeness proof for $\mathsf{PDL}$ by Kozen \& Parikh in~\cite{kozen1981elementary}. In the book~\cite{goldblatt1987logics}, Goldblatt applies the same procedure to several modal fixpoint logics, including $\mathsf{CTL}$. 
	
	In the recent paper~\cite{kikot2020completeness}, Kikot, Shapirovsky \& Zolin, prove a result of this kind that is relatively wide in scope. They show that if a basic modal logic $\mathsf{L}$ allows the method of filtration, then so does its expansion with the transitive closure modality. By iterating this procedure they show the same for the expansion of $\mathsf{L}$ by all modalities of $\mathsf{PDL}$. Subsequently, if the original basic modal logic $\mathsf{L}$ moreover is canonical, the completeness of this $\mathsf{PDL}$-expansion of $\mathsf{L}$ can be obtained by applying filtration to its canonical model.
	
	\paragraph{The continuous modal $\mu$-calculus}
	In this paper we consider the methods of filtration and canonical models for a specific fragment of $\mu \mathsf{ML}$, which is called the \textit{continuous modal $\mu$-calculus} and is denoted $\mu_c \mathsf{ML}$. In the paper~\cite{fontaine2008continuous}, Fontaine shows that there are two equivalent ways to define $\mu_c \mathsf{ML}$. First semantically, as the fragment of the modal $\mu$-calculus where the application of fixpoint operators is restricted to formulas whose functional interpretation is Scott-continuous, rather than merely monotone. And second syntactically, as the fragment where the modal operator $\lb$ and the fixpoint operator $\nu$ are not allowed to occur in the scope of a $\mu$-operator (and dually for the $\nu$-operator). To the best of our knowledge, the logic $\mu_c\mathsf{ML}$ was mentioned first
	in van Benthem~\cite{bent:moda06} under the name ‘$\omega\text{-}\mu$-calculus’. It is related, and perhaps equivalent in expressive power, to the logic 
	$\mathsf{CPDL}$ of concurrent propositional dynamic logic, cf. Carreiro~\cite[section 3.2]{carreiro2015fragments} for more information.
	
	There are at least two reasons why the continuous $\mu$-calculus is 
	an interesting logic; first, the continuity condition that is imposed on the 
	formation of fixpoint formulas ensures that the construction of a definable 
	fixpoint using its ordinal approximations will always be finished after $\omega$ 
	many steps.
	And second, in the same manner that the full $\mu$-calculus is the 
	bisimulation-invariant fragment of monadic second-order 
	logic~\cite{jani:expr96}, $\mu_c \mathsf{ML}$ has the same expressive power as 
	\emph{weak} monadic second-order logic, when it comes to bisimulation-invariant 
	properties~\cite{carr:powe20}.
	
	The goal of the present paper is to show that we can add two more desirable properties to this list: (i) the Filtration Theorem holds for $\mu_c \mathsf{ML}$ and (ii) completeness for sufficiently nice logics in the language of $\mu_c \mathsf{ML}$ can be proven using finitary canonical models. 
	
 	Since $\mu_c \mathsf{ML}$ is strictly more expressive than $\mathsf{PDL}$~\cite{fontaine2008continuous, carreiro2015fragments}, this is a proper generalisation of the aforementioned results from the paper~\cite{kikot2020completeness}. On the other hand, because the failure of filtration for $\mu \mathsf{ML}$ is witnessed by the formula $\mu x.\lb x$, the syntactic restrictions characterising $\mu_c \mathsf{ML}$ seem to be not only sufficient, but also necessary for filtration. This indicates that $\mu_c \mathsf{ML}$ might be positioned as a maximal filtration-allowing language between the basic modal language and the full language of the modal $\mu$-calculus. We leave it for future work to make this statement mathematically precise and to investigate its correctness.
 	
 	\paragraph{Overview of the paper} In Section 2 we define the syntax of the continuous modal $\mu$-calculus, the game semantics and other basic notions. In Section 3 we treat filtration. After giving the necessary definitions, we will use game-theoretic arguments to prove the Filtration Theorem for the language $\mu_c \mathsf{ML}$. As a corollary, we obtain the finite model property for this language interpreted over a wide range of model classes. In Section 4 we prove our completeness result, again using game-theoretic methods.
	
	Unlike for $\mathsf{PDL}$, there is no obvious way to construct a non-standard canonical model for $\mu_c \mathsf{ML}$.  Because of this, we define the finitary canonical model used in our completeness proof directly, instead of as some filtration of a non-standard canonical model. This causes Section 3 and Section 4 to contain some rather similar constructions and proofs. We leave it for future work to unify these two.
	\section{The continuous modal $\mu$-calculus}
	\paragraph{Syntax}  The continuous modal $\mu$-calculus will be defined using the syntactic characterisation given by Fontaine in~\cite{fontaine2008continuous}. We fix a countably infinite set $\mathsf{P}$ of propositional variables.
	\begin{definition}
		By simultaneous induction we define the following three languages.
		\begin{enumerate}[label = (\roman*)]
			\item The syntax $\mu_c \mathsf{ML}$ of the \textit{continuous modal $\mu$-calculus}:
			\[
			\varphi ::= p \ | \ \neg p \ | \ \varphi \lor \varphi \ | \ \varphi \land \varphi \ | \   \ld \varphi \ | \  \lb \varphi \ | \ \mu x .\varphi' | \ \nu x. \varphi''
			\]
			where $p, x \in \mathsf{P}$ and $\varphi' \in \mathsf{Con}_{\{x\}}(\mu_c \mathsf{ML})$, and
			$\varphi'' \in \mathsf{Cocon}_{\{x\}}(\mu_c \mathsf{ML})$.
			\item  For $\mathsf{X} \subseteq \mathsf{P}$, the fragment $\mathsf{Con}_{\mathsf{X}}(\mu_c \mathsf{ML})$ of $\mu_c \mathsf{ML}$-formulas that are \textit{continuous} in $\mathsf{X}$:
			\[
			\varphi ::= x  \mid \alpha \mid \varphi \lor \varphi \mid \varphi \land \varphi \mid  \ld \varphi \mid \mu y. \varphi'
			\]
			where $x \in \mathsf{X}$, $y \in \mathsf{P}$, $\alpha \in \mu_c \mathsf{ML}$ $\mathsf{X}$-free, and $\varphi' \in \mathsf{Con}_{\mathsf{X} \cup \{y\}} (\mu_c \mathsf{ML})$.
			\item   For $\mathsf{X} \subseteq \mathsf{P}$, the fragment $\mathsf{Cocon}_{\mathsf{X}}(\mu_c \mathsf{ML})$ of $\mu_c \mathsf{ML}$-formulas that are \textit{cocontinuous} in $\mathsf{X}$:
			\[
			\varphi ::= x  \mid \alpha \mid \varphi \lor \varphi \mid \varphi \land \varphi \mid  \lb \varphi \mid \nu y. \varphi'
			\]
			where $x \in \mathsf{X}$, $y \in \mathsf{P}$, $\alpha \in \mu_c \mathsf{ML}$ $\mathsf{X}$-free, and $\varphi' \in \mathsf{Cocon}_{\mathsf{X} \cup \{y\}} (\mu_c \mathsf{ML})$. \hfill \qedhere
		\end{enumerate}
	\end{definition}
	 If one of the above fragments is subscripted by a singleton $\{x\}$,  we will simply write $x$ instead. We will use \textit{formula} to refer to a $\mu_c \mathsf{ML}$-formula. We define the subformula relation $\unlhd$ and, for a given formula $\xi$, the sets $\textnormal{Sf}(\xi)$ of subformulas, $\textnormal{FV}(\xi)$ of free variables and $\textnormal{BV}(\xi)$ of bound variables of $\xi$ in the usual way. Given two formulas $\varphi, \psi$ and a propositional variable $x$, we define $\varphi[\psi/x]$ to be the result of replacing each free occurrence of $x$ in $\varphi$ by $\psi$. We will assume an implicit mechanism of $\alpha$-conversion in order to avoid the capture of free variables of $\psi$ by binders in $\varphi$ in the substitution $\varphi[\psi/x]$.
	 
	 We say that a formula is \textit{tidy} if the sets of its free and its bound variables are disjoint. A formula $\varphi$ is called \textit{clean} if, in addition, we can associate with each bound variable $x$, a unique fixpoint binder $\eta_x$ and a unique formula $\delta_x$ such that $\eta x  . \delta_x$ is a subformula of $\varphi$. In this case, if $\eta_x = \mu$ ($\eta_x = \nu$), the variable $x$ is said to be a \textit{$\mu$-variable} (\textit{$\nu$-variable}). We will sometimes denote by $\overline \eta$ the dual of $\eta$. Note that every subformula of a clean formula is itself clean. For convenience we will assume that every formula is tidy. Finally, we will use $\mathsf{ML}$ to refer to the basic modal language (over the set $\mathsf{P}$ of propositional variables).\newpage 
	\begin{definition}
		The \textit{FL-closure} of a set $\Phi$ of  $\mu_{c}\mathsf{ML}$-formulas is the least $\Psi \supseteq \Phi$ such that:
		\begin{enumerate}[label = (\roman*), noitemsep]
			\item If $\neg p \in \Psi$, then $p \in \Psi$;
			\item If $\varphi \circ \psi \in \Psi$ for $\circ \in \{\lor, \land\}$, then $\varphi, \psi \in \Psi$;
			\item If $\heartsuit \varphi \in \Psi$ for $\heartsuit \in \{\ld, \lb\}$, then $\varphi \in \Psi$;
			\item If $\eta x. \varphi \in \Psi$ for $\eta \in \{\mu, \nu\}$, then $\varphi[\eta x.\varphi / x] \in \Psi$.
		\end{enumerate}
		We write $Cl(\Phi)$ for the FL-closure of $\Phi$ and say that $\Phi$ is \textit{FL-closed} if $Cl(\Phi) = \Phi$. If $\Phi = \{\varphi\}$ is a singleton, we simply write $Cl(\varphi)$.
	\end{definition}
	It is a well-known fact that the closure of a finite set of formulas is finite. Note, moreover, that in the FL-closure of a set of tidy formulas, every formula is tidy.

	We say of a subformula  $\varphi \unlhd \xi$ that it is a \textit{free} subformula of $\xi$, and write $\varphi \unlhd_f \xi$, if $\varphi \in Cl(\xi)$. Equivalently, a subformula $\varphi \unlhd \xi$ is a free subformula of $\xi$ whenever every free variable of $\varphi$ is a free variable of $\xi$. 
	\paragraph{Algebraic semantics} As usual, formulas will be interpreted in Kripke models.
	\begin{definition}
		A \textit{Kripke frame} is a pair $(S, R)$ consisting of a set $S$ of \textit{states} together with an \textit{accessibility relation} $R \subseteq S \times S$. A \textit{Kripke model} is a triple $(S, R, V)$, where $(S, R)$ is a Kripke frame and $V : \mathsf{P} \rightarrow \mathcal{P}(S)$ a \textit{valuation function}.
	\end{definition}
	Given some accessibility relation $R$, we often write $sRt$ instead of $(s, t) \in R$. The algebraic semantics of the continuous $\mu$-calculus extends that of the basic modal language. Given a valuation function $V : \mathsf{P} \rightarrow \mathcal{P}(S)$, we write $V[x \mapsto X]$ for the function given by $V[x \mapsto X](x) = X$ and $V[x \mapsto X](y) = V(y)$ for $y \not= x$.
	\begin{definition}
		We define for every formula $\varphi$ its \textit{meaning} $[\![\varphi]\!]^\mathbb{S} \subseteq S$ in any model $\mathbb{S} = (S, R, V)$ by the following induction on formulas:
	\begin{align*}
			[\![\mu x .\varphi]\!]^\mathbb{S} &:= \bigcap \{X \subseteq S : [\![\varphi]\!]^{\mathbb{S}[x \mapsto X]} \subseteq X\} \\
			[\![\nu x .\varphi]\!]^\mathbb{S} &:= \bigcup \{X \subseteq S : X \subseteq [\![\varphi]\!]^{\mathbb{S}[x \mapsto X]}\} 
		\end{align*}
	and the propositional and modal cases are as usual.
	\end{definition}
		
	We say that $\xi$ is \textit{satisfied} at a state $s$ of the model $\mathbb{S}$, and write $\mathbb{S}, s \Vdash \xi$ whenever $s \in [\![\xi]\!]^\mathbb{S}$. As usual, we say that $\xi$ is \textit{valid} in $\mathbb{S}$, written $\mathbb{S} \models \xi$, whenever $\xi$ is satisfied at every state $s$ of $\mathbb{S}$, and \textit{valid} in the frame $(S, R)$, written $(S, R) \models \xi$, whenever $(S, R, V) \models \xi$ for every valuation function $V : \mathsf{P} \rightarrow \mathcal{P}(S)$. 
	
Two formulas are called \textit{equivalent} whenever they have the same meaning in every Kripke model. It easy to see that every formula has an equivalent alphabetic variant which is clean.
	
	\paragraph{Game semantics} A well-known equivalent characterisation of the meaning of a formula uses the formalism of infinite games. We assume familiarity with this kind of games.
	\begin{definition}
		Given a clean formula $\xi$, we define the \textit{dependency order} $<_\xi$ on $\text{BV}(\xi)$ as the least strict partial order such that $x <_\xi y$ whenever $\delta_x \lhd \delta_y$ and $y \lhd \delta_x$.
	\end{definition}
	Note that for formulas of the continuous $\mu$-calculus $x <_\xi y$ implies that $x$ is a $\mu$-variable if and only if $y$ is a $\mu$-variable. In other words, the continuous modal $\mu$-calculus is \textit{alternation free}.
	\begin{definition}
		Let $\xi$ be a clean formula and let $\mathbb{S} = (S, R, V)$ be a Kripke model. The \textit{evaluation game} $\mathcal{E}(\xi, \mathbb{S})$ takes positions in $\textnormal{Sf}(\xi) \times S$ and has the following ownership function and admissible moves.
	\begin{center}
		\begin{tabular}{| c l | c | c| }
			\hline
			Position && Player & Admissible moves \\ \hline
			$(\varphi_1 \lor \varphi_2, s)$ && $\exists$ & $\{(\varphi_1, s), (\varphi_2, s)\}$ \\ 
			$(\varphi_1 \land \varphi_2, s)$ && $\forall$ & $\{(\varphi_1, s), (\varphi_2, s)\}$ \\  
			$(\ld \varphi, s)$ && $\exists$ & $\{(\varphi, t) : sRt\}$ \\  
			$(\lb \varphi, s)$ && $\forall$ & $\{(\varphi, t) : sRt\}$ \\  
			$(\eta x .\delta_x, s)$ && - & $\{(\delta_x, s)\}$ \\  
			$(x, s)$& with $x \in \textnormal{BV}(\xi)$ & - & $\{(\delta_x, s)\}$ \\
			$(p, s)$& with $p \in \textnormal{FV}(\xi)$ and $s \in V(p)$ & $\forall$ & $\emptyset$ \\
			$(\neg p, s)$& with $p \in \textnormal{FV}(\xi)$ and $s \in V(p)$ & $\exists$ & $\emptyset$ \\
			$(p, s)$& with $p \in \textnormal{FV}(\xi)$ and $s \not \in V(p)$ & $\exists$ & $\emptyset$ \\
			$(\neg p, s)$& with $p \in \textnormal{FV}(\xi)$ and $s \not \in V(p)$ & $\forall$ & $\emptyset$ \\
			\hline
		\end{tabular}
	\end{center}
	For $\gamma$ a match in $\mathcal{E}(\xi, \mathbb{S})$, we denote the first position of $\gamma$ by $\mathsf{first}(\gamma)$ and, if $\gamma$ is finite, the last position by $\mathsf{last}(\gamma)$. A finite match $\gamma$ is won by one of the players whenever $\mathsf{last}(\gamma)$ is owned by its opponent and this opponent's set of admissible moves is empty (in this case the opponent is said to have gotten \textit{stuck}). An infinite match is won by $\exists$ ($\forall$) if the $<_\xi$-highest variable that is unfolded infinitely often is a $\nu$-variable (a $\mu$-variable). We write $(\varphi, s) \in \text{Win}_\exists(\mathcal{E}(\xi, \mathbb{S}))$ to denote that $\exists$ has a winning strategy in the game $\mathcal{E}(\xi, \mathbb{S})$ initialised at position $(\varphi, s)$. 
	\end{definition}
	The following lemma contains some basic facts about the course of play in evaluation games for $\mu_c \mathsf{ML}$. Items (1) and (2) hold because the continuous $\mu$-calculus is alternation free. Item (3) is specific to the continuous modal $\mu$-calculus, in the sense that it does not hold for the more expressive \textit{alternation free $\mu$-calculus} (see \textit{e.g.}~\cite{martivenemafocus} for a formal definition of this language).
	\begin{lemma} 
		\label{lem:conplay} Let $\mathbb{S}$ be a model and let $\xi$ be a clean formula.
		\begin{enumerate}
			\item In any infinite match of the game $\mathcal{E}(\xi, \mathbb{S})$, either all variables that are unfolded infinitely often are $\mu$-variables, or all are $\nu$-variables.
			\item If a match of the game $\mathcal{E}(\xi, \mathbb{S})$ progresses from a position $(s, \eta x . \delta)$ to a position $(t, \overline \eta y . \theta)$, then in between it must pass a position $(r, \varphi)$ with $\varphi \lhd_f \xi$.
			\item If a match of the game $\mathcal{E}(\xi, \mathbb{S})$ progresses from a position $(s, \mu x . \delta)$ to a position $(t, \lb \psi)$, then in between it must pass a position $(r, \varphi)$ with $\varphi \lhd_f \xi$. \hfill \qedhere
		\end{enumerate}
	\end{lemma}
	We say of an enumeration $\{x_1, \ldots x_n\}$ of $\text{BV}(\xi$) that it \textit{respects} the dependency order if $x_i <_\xi x_j$ implies $i < j$. Since any partial order can be extended to a linear order, every formula $\xi$ admits an enumeration of its bound variables that respects the dependency order. For the rest of this paper we fix such an enumeration of $\text{BV}(\xi)$ for every clean formula $\xi$. 
	\begin{definition}
		Let $\xi$ be a clean formula with $\text{BV}({\xi}) = \{x_1, \ldots, x_n\}$. For any subformula $\varphi \unlhd \xi$, we define its \textit{expansion} with respect to $\xi$ as:
		\[
		\exp_\xi(\varphi) := \varphi[\eta x_1\delta_{x_1}/x_1]\cdots[\eta x_n\delta_{x_n}/x_n]. \hfill \qedhere
		\]
	\end{definition}
	Note that when $\varphi \unlhd_f \xi$, it holds that $\exp_\xi (\varphi) = \varphi$. The following well-known theorem provides the central link between the algebraic and the game semantics. 
	\begin{theorem}
		\label{thm:gamalg}
		For any clean formula $\xi$ and subformula $\varphi \unlhd \xi$ it holds that:
		\[
		(\varphi, s) \in \text{Win}_\exists(\mathcal{E}(\xi, \mathbb{S})) \Leftrightarrow \mathbb{S}, s \Vdash \exp_\xi(\varphi).
		\]
		for any model $\mathbb{S}$ and state $s$ of $\mathbb{S}$.
	\end{theorem}
	In particular, for any clean formula $\xi$ and $\varphi \unlhd_f \xi$ we have $\mathbb{S}, s \Vdash \varphi$ if and only if $\exists$ has a winning strategy in the game $\mathcal{E}(\xi, \mathbb{S})$ initialised at the position $(\varphi, s)$. 
	
	Another useful fact, originally provided by Dexter Kozen in~\cite{kozen1983results}, is the following.
	\begin{proposition}
		For any clean formula $\xi$:
		\[
		Cl(\xi) = \{\exp_\xi(\varphi) : \varphi \unlhd \xi\}.
		\]
	\end{proposition}

	\paragraph{Axiomatisation} We give an axiomatisation of the continuous modal $\mu$-calculus based on an axiomatisation introduced by Dexter Kozen for the full modal $\mu$-calculus in~\cite{kozen1983results}.
	\begin{definition}
		The logic $\mu_c \mathsf{K}$ is the least logic containing the following axioms and closed under the following rules.\footnote{Because we have defined $\mu_c \mathsf{ML}$ in negation normal form, we formally also need to add the dual version of each axiom and rule. Moreover, we should have rules expressing that $\ld$ and $\lb$ and, respectively, $\mu$ and $\nu$ are duals. For reasons of space and clarity we omit these technical details.}\\
		\noindent \textbf{Axioms.}
		\begin{enumerate}[noitemsep, topsep = 0pt]
			\item A complete set of axioms for classical propositional logic.
			\item Normality: $\neg \ld \bot$.
			\item Additivity: $\ld (p \lor q) \leftrightarrow (\ld p \lor \ld q)$.
			\item For every $\varphi \in \mathsf{Con}_x(\mu_{c}\mathsf{ML})$, the prefixpoint axiom: 
			\[\varphi[\mu x . \varphi / x] \rightarrow \mu x .  \varphi.\]
		\end{enumerate} 
		\textbf{Rules.}
		\begin{enumerate}[noitemsep, topsep = 0pt]
			\item Modus Ponens: from $\varphi \rightarrow \psi$ and $\varphi$, derive $\psi$. 
			\item Monotonicity: from $\varphi \rightarrow \psi$, derive $\ld \varphi \rightarrow \ld \psi$. 
			\item Uniform Substitution: from $\varphi$, derive $\varphi[\psi/x]$. 
			\item The least prefixpoint rule: from $ \varphi[\gamma / x] \rightarrow \gamma$ with $\varphi \in \mathsf{Con}_x(\mu_{c}\mathsf{ML})$, derive $\mu x. \varphi \rightarrow \gamma$. \qedhere
		\end{enumerate}
	\end{definition}
	We will consider axiomatic extensions of $\mu_c \mathsf{K}$ that are closed under the rules above. We will use \textit{$\mu_c$-logic} to refer to such an extension. The term \textit{logic} will be used to refer to any normal modal logic. If $\mathsf{L}$ is a logic in the basic modal language, we use $\mu_c$-$\mathsf{L}$ to denote the least $\mu_c$-logic containing $\mathsf{L}$. Moreover, we will use $\mathsf{Mod}(\mathsf{L})$ ($\mathsf{Fr}(\mathsf{L})$) to denote the class of models (frames) on which every formula in $\mathsf{L}$ is valid. If $(S, R, V)$ belongs to $\mathsf{Mod}(\mathsf{L})$ ($(S, R)$ belongs to $\mathsf{Fr}(\mathsf{L})$) we say that $(S, R, V)$ is an $\mathsf{L}$-model ($(S, R)$ is an $\mathsf{L}$-frame) and write $(S, R, V) \models \mathsf{L}$ ($(S, R) \models \mathsf{L}$).
	\section{Filtration}
	Filtration is a well-known method in the theory of basic modal logic. In this section we define filtration and related notions for the continuous modal $\mu$-calculus and show that some of their most important properties transfer to this more expressive language.
	\paragraph{Filtration} 
	\begin{definition}
		\label{def:filtration}
		Let $\mathbb{S} = (S, R, V)$ be a Kripke model and let $\Sigma$ be a finite and FL-closed set of formulas. Let $\sim_\Sigma^\mathbb{S}$ be the equivalence relation given by:
		\[
		s \sim^\mathbb{S}_\Sigma s' \textnormal{ if and only if } \mathbb{S}, s \Vdash \varphi \Leftrightarrow \mathbb{S}, s' \Vdash \varphi \textnormal{ for all $\varphi \in \Sigma$.}
		\]
		A $\Sigma$-\textit{filtration} of $\mathbb{S}$ \textit{through} $\Sigma$ is a model $\mathbb{S}^\Sigma = (S^\Sigma, R^\Sigma, V^\Sigma)$ such that:
		\begin{enumerate}[label = (\roman*)]
			\item $S^\Sigma = S/{\sim_\Sigma^\mathbb{S}}$
			\item $R^\mathsf{min} \subseteq R^\Sigma \subseteq R^\mathsf{max}$;
			\item $V^\Sigma(p) = \{\overline s : s \Vdash p\}$ for every propositional variable $p \in \Sigma$.
		\end{enumerate}
		where:
		\begin{align*}
			R^\mathsf{min} &:= \{(\overline s, \overline t) : \textnormal{there are $s' \sim_\Sigma^\mathbb{S} s$ and $t' \sim_\Sigma^\mathbb{S} t$ such that } R s't'\}, \\
			R^\mathsf{max} &:= \{(\overline s, \overline t) : \textnormal{for all $\lb \varphi \in \Sigma$; if $s \Vdash \lb \varphi$, then $t \Vdash \varphi$}\}. 
		\end{align*}
	where $\overline s$ denotes the equivalence class with representative $s$.
	\end{definition}
	The relation $R^\mathsf{min}$ will be called the \textit{finest filtration} and the relation $R^\mathsf{max}$ the \textit{coarsest}.
	\paragraph{Filtration Theorem for the continuous modal $\mu$-calculus}
If $f$ is a strategy for the player $\exists$ ($\forall$) in a game $\mathcal{G}$, we say of a (possibly infinite) $\mathcal{G}$-match $\gamma$ that it is $f$-\textit{guided} whenever every choice made by $\exists$ ($\forall$) in the match $\gamma$ is the choice dictated by the strategy $f$.
\begin{theorem}[Filtration Theorem]
	\label{thm:filtrationthm}
	Let $\Sigma$ be a finite and FL-closed set of formulas and let $\mathbb{S} = (S, R, V)$ be a Kripke model. For every filtration $\overline{\mathbb{S}} = (\overline S, \overline R, \overline V)$ of $\mathbb{S}$ through $\Sigma$ it holds that
	\[
	\mathbb{S}, s \Vdash \xi \Leftrightarrow \overline{ \mathbb{S}}, \overline s \Vdash \xi,
	\] 
	for every clean formula $\xi \in \Sigma$.
\end{theorem}
\begin{proof}
	Because negation is definable in our language, it suffices to prove the implication in just one direction, which in our case will be the direction $\Rightarrow$. Throughout this proof we will write $\mathcal{G}$ for the game $\mathcal{E}(\xi, \mathbb{S})$ and $\overline{\mathcal{G}}$ for the game $\mathcal{E}(\xi, \overline {\mathbb{S}})$. As hypothesis we assume that $\exists$ has a winning strategy $f$ in the game $\mathcal{G}$ initialised at position $(\xi, s)$; we wish to show that $(\xi, \overline s) \in \text{Win}_\exists(\overline{\mathcal{G}})$.
	
	The main idea of the proof is to obtain a winning strategy for $\exists$ in $\overline{\mathcal{G}}$ by playing a `shadow match' in $\mathcal{G}$. That is, we will simulate in $\mathcal{G}$ every move played by $\forall$ in our $\overline{\mathcal{G}}$-match, and, to determine a move for $\exists$ in $\overline{\mathcal{G}}$, we copy the move dictated in $\mathcal{G}$ by the strategy $f$. If we manage to do this, then whenever the match in $\overline{\mathcal{G}}$ is at some position $(\varphi, \overline s)$, the shadow match in $\mathcal{G}$ will be at a position $(\varphi, s)$ (note that this is indeed the case for the initial positions). It turns out that this works well for all positions, except those of the form $(\lb \varphi, \overline s)$. At those positions, a problem arises when $\forall$ chooses a position $(\varphi, \overline t)$ such that $\overline s \overline R \overline t$, but not $sRt$. This move by $\forall$ in $\overline{\mathcal{G}}$ can then not be simulated in the shadow match, because $(\varphi, t)$ is not an admissible move for $\forall$ in $\mathcal{G}$. However, using the fact that $\overline R \subseteq R^\mathsf{max}$, we will be able to show that if $\overline s \overline R \overline t$ and $(\lb \varphi, s) \in \text{Win}_\exists(\mathcal{G})$, then $(\varphi, t) \in \text{Win}_\exists(\mathcal{G})$. We will use this to initiate a new shadow match in $\mathcal{G}$ whenever encounter a position of the form $(\lb \varphi, \overline s)$. A key observation will be that we only need to initiate a new shadow match at most finitely many times, because formulas of the form $\lb \varphi$ do not occur in the scope of least fixed point operators in the language $\mu_c \mathsf{ML}$.
	
	More formally, we say that for $I \in \omega \cup \{\omega\}$, a $\overline{\mathcal{G}}$-match $\overline{\gamma} = (\varphi_i, \overline{t_i})_{i \in I}$ is \textit{linked} to some $\mathcal{G}$-match $\gamma = (\psi_i, s_i)_{i \in I}$ whenever for every $i \in I$ it holds that $\varphi_i = \psi_i$ and $\overline{s_i} = \overline{t_i}$. Moreover, we say that $\overline \gamma$ \textit{follows} $\gamma$ whenever some final segment of $\overline \gamma$ is linked to $\gamma$.
	
	\begin{quote}
		\emph{Claim}. Let $\overline \gamma$ be a finite $\overline{\mathcal{G}}$-match that follows some $f$-guided $\mathcal{G}$-match $\gamma$, where $f$ is a winning strategy for $\mathcal{G}$ initialised at $\mathsf{first}(\gamma)$. Then:
		\begin{itemize}
			\item If the formula in  $\mathsf{last}(\gamma)$ is not of the form $\lb \theta$, then $\exists$ can ensure that after the next round in $\overline{\mathcal{G}}$, there is some admissible move $(\psi_{n+1}, t_{n+1})$ in $\mathcal{G}$ such that the resulting $\overline{\mathcal{G}}$-match follows the $\mathcal{G}$-match $\gamma \cdot (\psi_{n+1}, t_{n+1})$ and the latter remains $f$-guided.
			
			\item If the formula in $\mathsf{last}(\gamma)$ is of the form $\lb \theta$, then $\exists$ can at least ensure that after the next round, the resulting $\overline{\mathcal{G}}$-match follows a new match $(\theta, t_0)$ for which $\exists$ has a winning strategy.
		\end{itemize}
	\end{quote}
	The above claim is proven by a case distinction on the main connective of the formula in $\mathsf{last}(\gamma)$. We treat the most difficult cases of $\ld$ and $\lb$, leaving the rest to the reader.
	
	Suppose $\mathsf{last}(\gamma)$ is of the form $(\ld \theta, t_n)$. Let $(\theta, t_{n + 1})$ be the next move instructed by the assumed winning strategy $f$. Then $t_n R t_{n+1}$ and thus, because $\overline R \subseteq R^\mathsf{min}$ and $s_n \sim t_n$, we have $\overline{s_n} R \overline{t_{n+1}}$. Therefore $\exists$ can simply choose the position $(\theta, \overline {t_{n+1}})$.
	
	If $\mathsf{last}(\gamma)$ is of the form $(\lb \theta, t_n)$, consider the move $(\theta, \overline{s_{n+1}})$ chosen by $\forall$ in $\overline{\mathcal{G}}$. We have,
	\begin{align*}
		(\lb \theta, t_n) \in \text{Win}_\exists(\mathcal{G}) &\Rightarrow \mathbb{S}, t_n \Vdash \exp_\xi (\lb \theta)  & \text{(Theorem \ref{thm:gamalg})}\\
		&\Rightarrow \mathbb{S}, s_n \Vdash \exp_\xi(\lb \theta) & \text{($\exp_\xi(\lb \theta) \in \Sigma$ and $s_n \sim t_n$)} \\
		&\Rightarrow \mathbb{S}, s_n \Vdash \lb \exp_\xi(\theta) & \text{(Definition of $\exp$)} \\
		&\Rightarrow \mathbb{S}, s_{n+1} \Vdash \exp_\xi(\theta) & \text{$(\lb\exp_\xi( \theta) \in \Sigma$ and $\overline{s_n} R^\mathsf{max} \overline{s_{n+1}}$)} \\
		&\Rightarrow (\theta, s_{n +1}) \in \text{Win}_\exists(\mathcal{G}). & \text{(Theorem \ref{thm:gamalg})}
	\end{align*}
	Thus we may choose $(\theta, s_{n+1})$ as the new match that is followed by $\overline{\gamma} \cdot (\theta, \overline{s_{n+1}})$.
	
	Using the fact that $(\xi, s)$ is linked to $(\xi, \overline s)$ as induction base, and the above claim as induction step, we obtain a strategy $g$ for $\exists$ in $\overline{\mathcal{G}}$ initialised at $(\xi, \overline s)$. We claim that $g$ is a winning strategy. Indeed, if a $g$-guided match $\overline{\gamma}$ ends in finitely many steps, then either $\forall$ got stuck on a formula of the form $\lb \theta$, or the final position is of the form $(\pi, \overline{t})$ for some $\pi = p, \neg p \in \text{FV}(\xi)$. Without loss of generality, suppose $\pi = p$. By construction $\overline \gamma$ follows a $\mathcal{G}$-match $\gamma$ such that $\mathsf{last}(\gamma) = (\pi, r) \in \text{Win}_\exists(\mathcal{G})$ for some state $r \in S$ with $\overline r = \overline t$. But this means that $r \in V(p)$ and thus, since $p \in \Sigma$, also $\overline t \in \overline V(p)$. Hence $\exists$ indeed wins the match $\overline \gamma$.
	
	If a $g$-guided match $\overline \gamma$ lasts infinitely long, then by item (1) of Lemma \ref{lem:conplay}, there must be some point after which either only $\mu$-variables, or only $\nu$-variables, are unfolded. In the latter case the match is indeed winning for $\exists$. We will now argue that this is the only possibility, because the former case cannot occur. The reason is that if from some point on in $\overline \gamma$ only $\mu$-variables are unfolded, then the syntax of $\mu_c \mathsf{ML}$ dictates that from some point on in $\overline \gamma$ no formula of the form $\lb \theta$ will occur. By construction, this means that the infinite $\overline{ \mathcal{G}}$-match $\overline \gamma$ follows an infinite $\mathcal{G}$-match $\gamma$ which is guided by a winning strategy for $\exists$. But this is a contradiction, because the match $\gamma$, by the fact that it is linked to an infinite final segment of $\overline \gamma$, contains infinitely many $\mu$-unfoldings.
\end{proof}
Note that the above argument would not go through for the alternation free $\mu$-calculus, since we would no longer be able to guarantee that we create at most finitely many shadow matches in the case of infinitely many $\mu$-unfoldings. A well-known counterexample to the Filtration Theorem for the alternation free $\mu$-calculus is the formula $\mu x.\lb x$.
	\paragraph{Admissibility of filtration} 
	Having established that filtrations preserve satisfaction of $\mu_c \mathsf{ML}$-formulas, we will now investigate to which classes of models filtration can be applied.
	\begin{definition}
		A class of models $\mathcal{M}$ is said to \textit{admit filtration} with respect to a language $\mathsf{D}$ if for every model $\mathbb{S}$ in $\mathcal{M}$ and every finite FL-closed set of $\mathsf{D}$-formulas $\Sigma$, the class $\mathcal{M}$ contains a filtration of $\mathbb{S}$ through $\Sigma$. A class of frames $\mathcal{F}$ is said to \textit{admit filtration} if the class of models $\{(S, R, V) : (S, R) \in \mathcal{F}\}$ does.
	\end{definition}
	One might expect that admitting filtration with respect to the basic modal language is a weaker property than admitting filtration with respect to a proper extension of the language. However, for the language $\mu_c \mathsf{ML}$ it turns out that this is not the case, at least for those classes of models that are determined by some logic.
	
	We will show this by making use of the following technical sufficient condition.
	\begin{lemma}
		\label{lem:sufcon}
		Let $\mathcal{M}$ be a class of models that admits filtration wrt $\mathsf{ML}$. Suppose that for every model $\mathbb{S} := (S, R, V) \in \mathcal{M}$ and finite FL-closed set $\Sigma \subset \mu_c \mathsf{ML}$, there is a valuation $V' : \mathsf{P} \rightarrow \mathcal{P}(S)$ and a translation $\tau : \Sigma \rightarrow \mathsf{ML}$ such that:
		\begin{multicols}{2}	
		\begin{enumerate}
			\item $V'(p) = p$ for all $p \in \mathsf{P}$ occurring in $\Sigma$;
			\item The model $\mathbb{S}' := (S, R, V')$ belongs to $\mathcal{M}$.
			\item The set $\tau[\Sigma] \subset \mathsf{ML}$ is FL-closed.
			\item The translation $\tau$ commutes with $\lb$, \textit{i.e.} $\tau(\lb \varphi) = \lb \tau (\varphi)$ for all $\lb \varphi \in \Sigma$.
			\item For every $\xi \in \Sigma$ and $s \in S$ it holds that: $\mathbb{S}, s \Vdash \xi \Leftrightarrow \mathbb{S}', s \Vdash \tau(\xi)$.
		\end{enumerate}
	\end{multicols}
	Then $\mathcal{M}$ admits filtration wrt $\mu_c \mathsf{ML}$.
	\end{lemma}
	\begin{proof}
	 Using conditions (2) and (3) and the assumption that $\mathcal{M}$ admits filtration with respect to $\mathsf{ML}$, there is a filtration $\mathbb{S}^{\tau[\Sigma]} \in \mathcal{M}$ of $\mathbb{S}'$ through $\tau[\Sigma]$. We claim that $\mathbb{S}^{\tau[\Sigma]}$ simultaneously is a filtration of $\mathbb{S}$ through $\Sigma$. 
	 
	 By assumption (5), the equivalence relations $\sim^\mathbb{S}_\Sigma$ and $\sim^{\mathbb{S}'}_{\tau[\Sigma]}$ on $S$ coincide. From this we obtain condition (i) of Definition \ref{def:filtration}, as well as the first inclusion of condition (ii). For the second inclusion, suppose that $\overline sR^{\tau[\Sigma]} \overline t$ and $\mathbb{S},s \Vdash \lb \varphi$ for some $\lb \varphi \in \Sigma$. We must show that $\mathbb{S}, t \Vdash \varphi$. By assumption (5), we have $\mathbb{S}', s \Vdash \tau(\lb \varphi)$ and thus, by assumption (4), also $\mathbb{S}', s \Vdash \lb \tau(\varphi)$. Since $R^{\tau[\Sigma]}$ is contained in the coarsest filtration of $\mathbb{S}'$ through $\tau[\Sigma]$ and $\lb \tau(\varphi) = \tau(\lb \varphi) \in \tau[\Sigma]$, we obtain $\mathbb{S}', t \Vdash \tau(\varphi)$. Applying the other direction of assumption (5), we obtain $\mathbb{S}, t \Vdash \varphi$, as required. Finally, condition (iii) follows directly from assumption (1).
	\end{proof}
	 The proof of the following lemma resembles that of Theorem 3.8 in~\cite{kikot2020completeness}.
	\begin{lemma}
		\label{lem:admitbasicadmitmucforlogics}
	For any logic $\mathsf{L}$, the class $\mathsf{Mod}(\mathsf{L})$ admits filtration wrt $\mathsf{ML}$ iff it admits filtration wrt  $\mu_c \mathsf{ML}$.
	\end{lemma}
	\begin{proof}
		The implication from right to left is trivial. For the other direction we will use Lemma \ref{lem:sufcon}. Let $\mathbb{S} = (S, R, V)$ be a model such that $\mathbb{S} \models \mathsf{L}$ and let $\Sigma$ be a finite FL-closed set of $\mu_c \mathsf{ML}$-formulas. Let $\varphi_1, \ldots, \varphi_n$ be an enumeration of formulas of the form $\eta x .\psi$ in $\Sigma$. For every such formula $\varphi_i$, we pick a unique propositional variable $p_i$ not occurring in $\Sigma$.

		We define the following alternative valuation $V' : \mathsf{P} \rightarrow \mathcal{P}(S)$.
		\[
		V'(p) := \begin{cases} [\![\varphi_i]\!]^\mathbb{S} &\text{ if $p = p_{i}$ for some $\varphi_i \in \Sigma$;} \\
			V(p) &\text{otherwise,}
		\end{cases}
		\]
		and define $\mathbb{S}' := (S, R, V')$. A straightforward induction on formulas now shows that for every formula $\xi \in \mu_c \mathsf{ML}$ and state $s \in S$:
		\begin{equation}
			\label{eq:invtau}
			\mathbb{S}', s \Vdash \xi \Leftrightarrow \mathbb{S}', s \Vdash \xi[\varphi_1/p_1]\cdots[\varphi_n/p_n].
		\end{equation}
		We claim that $\mathbb{S}' \in \mathsf{Mod}(\mathsf{L})$. Indeed, we have
		\begin{align*}
			\xi \in \mathsf{L} & \Rightarrow \xi[\varphi_1/p_{1}]\cdots[\varphi_n / p_{n}]  \in \mathsf{L} & \text{($\mathsf{L}$ is closed under uniform substitution)}\\
			&\Rightarrow \mathbb{S} \models \xi[\varphi_1/p_{1}]\cdots[\varphi_n / p_{n}] & \text{($\mathbb{S} \models \mathsf{L}$)}\\
			&\Rightarrow \mathbb{S}' \models \xi[\varphi_1/p_{1}]\cdots[\varphi_n / p_{n}] & \text{($V$ and $V'$ agree on all relevant propositional variables)} \\
			&\Rightarrow \mathbb{S}' \models \xi & \text{(by (\ref{eq:invtau}) from right to left)} 
		\end{align*}

		Now let the translation $\tau : \Sigma \rightarrow \mathsf{\mathsf{ML}}$ be the translation that commutes with all propositional and modal symbols, and acts on fixpoint operators in the following way:
		\[
		\tau(\eta x .\psi) := p_i \text{ where $\eta x . \psi = \varphi_i$.} \\
		\]
		We leave it to the reader to verify that $\tau[\Sigma]$ is FL-closed. Finally, another straightforward induction shows that for every formula $\xi \in \Sigma$ and state $s \in \mathbb{S}$:
		\begin{equation*}
			\mathbb{S}, s \Vdash \xi \Leftrightarrow \mathbb{S}',s \Vdash \tau(\xi).
		\end{equation*}
	This finishes the proof, for all conditions of Lemma \ref{lem:sufcon} are met. 
	\end{proof}
	Note that the above proof does not rely on any specific properties of the language $\mu_c \mathsf{ML}$. In fact, it could also have been carried out for the full language $\mu \mathsf{ML}$ of the modal $\mu$-calculus.  As a corollary, we obtain the finite model property.
	\begin{corollary}[Finite Model Property]
	Let $\mathsf{L}$ be a logic such that $\mathsf{Mod}(\mathsf{L})$ admits filtration with respect to $\mathsf{ML}$, and let $\phi$ be a formula of the continuous $\mu$-calculus. Then $\phi$ is valid in every $\mathsf{L}$-model if and only if $\phi$ is valid in every finite $\mathsf{L}$-model.
	\end{corollary}
	\begin{proof}
	Let $\varphi$ be a formula such that $\mathbb{S} \not \models \varphi$ for some $\mathbb{S} \models \mathsf{L}$. Without loss of generality, we may assume that $\varphi$ is clean. Letting $\Sigma = Cl(\varphi)$, there is, by Lemma \ref{lem:admitbasicadmitmucforlogics} and the fact that $\mathsf{Mod}(\mathsf{L})$ admits filtration, a filtration $\mathbb{S}^\Sigma$ of $\mathbb{S}$ through $\Sigma$ such that $\mathbb{S}^\Sigma \models \mathsf{L}$. Observe that number of states of $\mathbb{S}^\Sigma$ is at most $2^{|Cl(\varphi)|}$ and thus finite. By Theorem \ref{thm:filtrationthm}, it holds that $\mathbb{S}^\Sigma \not \models \varphi$, as required.
\end{proof}
For instance, since the class of symmetric models is the class of $\mathsf{KB}$-models, the continuous modal $\mu$-calculus has the finite model property over this class.
	\section{Canonical completeness}
	In this section we prove our completeness result. In the first paragraph, we will define the finitary canonical models of an arbitrary $\mu_c$-logic $\mathsf{L}$ and prove the Truth Lemma. In the second paragraph we will show that a finitary canonical model can be obtained for the logic $\mu_c$-$\mathsf{L}$, where $\mathsf{L}$ is any canonical basic modal logic such that $\mathsf{Fr}(\mathsf{L})$ admits filtration. As a direct consequence we obtain that $\mu_c$-$ \mathsf{L}$ is sound and complete with respect to $\mathsf{Fr}(\mathsf{L})$.
	\paragraph{Finitary canonical models}
	For the entirety of this paragraph we fix an arbitrary $\mu_c$-logic $\mathsf{L}$. We define the negation operator $\sim : \mu_c \mathsf{ML} \rightarrow \mu_c \mathsf{ML}$ in the usual way. In particular, that means that we define ${\sim} \eta x .\varphi := \overline \eta x .{\sim} \varphi[\neg x /x]$. We leave it to the reader to verify that $\mathsf{L} \vdash ({\sim} \varphi \land \varphi) \leftrightarrow \bot$ and $\mathsf{L} \vdash ({\sim} \varphi \lor \varphi) \leftrightarrow \top$. 
	\begin{definition}
		Let $\Sigma$ be a set of formulas. If for all $\varphi \in \Sigma$ it holds that ${\sim} \varphi \in \Sigma$, then $\Sigma$ is said to be $\sim$-closed.
	\end{definition}
	We say that $\Sigma$ is ${\sim}$FL-closed if it is both FL-closed and ${\sim}$-closed. Note that for every finite set of $\mu_c\mathsf{ML}$-formulas, the $\sim$-closure of its FL-closure is a finite ${\sim}$FL-closed extension. 
	
	A set $\Gamma$ of formulas is said to be $\mathsf{L}$-\textit{inconsistent} if $\mathsf{L} \vdash (\gamma_1 \land \ldots \land \gamma_n) \rightarrow \bot$ for some $\gamma_1, \ldots, \gamma_n \in \Gamma$. We say of a formula $\varphi$ that it is $\mathsf{L}$-inconsistent whenever $\{\varphi\}$ is. 
	\begin{definition}
		A set of formulas $\Gamma$ is called \textit{maximally $\mathsf{L}$-consistent} if it is consistent and maximal in that respect, \textit{i.e.} for every other set of formulas $\Gamma'$:
		\[
		\text{If $\Gamma \subset \Gamma'$, then $\Gamma'$ is $\mathsf{L}$-inconsistent.} \hfill \qedhere
		\]
	\end{definition}
	By a standard argument it can be shown that every $\mathsf{L}$-consistent set of formulas has a maximally $\mathsf{L}$-consistent extension. The proof of the following lemma is also standard and left to the reader.
	\begin{lemma}
		Let $\Gamma$ be a maximally $\mathsf{L}$-consistent set. Then:
		\begin{enumerate}[label = (\roman*), itemsep = 1pt, topsep = 1.5pt]
			\item If $\mathsf{L} \vdash \varphi$, then $\varphi \in \Gamma$;
			\item ${\sim} \varphi \in \Gamma$ if and only $\varphi \not \in \Gamma$;
			\item $\varphi \lor \psi \in \Gamma$ if and only $\varphi \in \Gamma$ or $\psi \in \Gamma$;
			\item $\mu x . \varphi \in \Gamma$ if and only if $\varphi[\mu x . \varphi/x] \in \Gamma$.
		\end{enumerate}
	\end{lemma}
	\begin{definition}
		Let $\Sigma$ be a finite ${\sim}$FL-closed set of formulas. A \textit{model over $\Sigma$ with respect to $\mathsf{L}$} is any model $(S, R, V)$ such that: 
		\begin{itemize}[noitemsep, topsep=1pt]
			\item $S = \{\Gamma \cap \Sigma : \text{$\Gamma$ is maximally $\mathsf{L}$-consistent}\}$.
			\item $R^\mathsf{min} \subseteq R \subseteq R^\mathsf{max}$, where:
			\begin{align*}
				AR^\mathsf{min}B &:\Leftrightarrow \bigwedge A \land \ld \bigwedge B \text{ is $\mathsf{L}$-consistent}  \\
				AR^\mathsf{max}B &:\Leftrightarrow \text{ for all } \lb \varphi \in \Sigma: \ \lb \varphi \in A \Rightarrow \varphi \in B.
			\end{align*}
			\item $V(p) = \{s  \in S : p \in s\}$ for all $p \in \Sigma$. \hfill \qedhere
		\end{itemize}
	\end{definition}
	For $A$ some finite set of formulas, we will usually write $\psi_A$ for the conjunction $\bigwedge A$. In the following we will assume a fixed model over some finite and ${\sim}$FL-closed set $\Sigma$ with respect to $\mathsf{L}$, which will be denoted by $\mathbb{S}^\Sigma = (S^\Sigma, R^\Sigma, V^\Sigma)$. We will often drop the superscript
	$\Sigma$'s and $\mathbb{S}$'s. Moreover, if in the following we refer to provability or consistency, this will be tacitly assumed to be in the logic $\mathsf{L}$.
	
	The following existence lemma is standard in the context of (finitary) canonical models for modal logics.
	\begin{lemma}
		For any formula $\varphi \in \mu_c\mathsf{ML}$ and state $A \in S$: 
		\begin{center}
			$\psi_A \land \ld \varphi$ is consistent if and only if $\psi_B \land \varphi$ is consistent for some $A R B$.
		\end{center}
	\end{lemma}
	\noindent In particular, it follows that for all $\ld \varphi \in \Sigma$ we have $\ld \varphi \in A$ if and only if $\varphi \in B$ for some $AR B$. The following lemma follows from the fact that $\Sigma$ is $\sim$-closed.
	\begin{lemma}
		For every $A, B \in S$ it holds that $\psi_A \land \psi_B$ is consistent iff $A = B$. 
	\end{lemma}
	Given a finite collection $U$ of finite sets of formulas, we write $\psi_U$ for the disjunction of all $\psi_X$ for $X \in U$, \textit{i.e.}
	\[
	\psi_U = \bigvee_{X \in U} \psi_X.
	\]
	Note that by the previous lemma, for any $U \subseteq S$ and $A \in S$, the formula $\psi_U \land \psi_A$ is consistent if and only if $A \in U$. 
	
	We wish to prove the following lemma.
	\begin{lemma}(Truth Lemma)
		\label{lem:truth}
		If $A \in S$ and $\xi \in \Sigma$ is clean, then
		\begin{equation}
			\xi \in A \Rightarrow A \in [\![\xi]\!]. \tag{T} \label{eq:T}
		\end{equation}
	\end{lemma}
	We shall prove this by a double induction on formulas, of which the inner induction is captured by Lemma \ref{lem:innerind}. 
	\begin{definition}
		Let $\xi$ a formula with $\text{BV}(\xi) = \{x_1, \ldots, x_n\}$. We define the \textit{name-expansion} $n$-$exp^\mathbb{S}_\xi(\varphi)$ of a subformula $\varphi$ of $\xi$ in $\mathbb{S}$ as follows:
		\begin{equation*}
			n\text{-}exp^\mathbb{S}_\xi(\varphi) := \varphi[\psi_{U_1}/x_1]\cdots[\psi_{U_n}/x_n], 
		\end{equation*}
		where $U_i := {[\![\exp_\xi(\delta_{x_i})]\!]^\mathbb{S}}$ for every $1 \leq i \leq n$.
	\end{definition}
	Whenever clear from context, we drop the subscript and superscript from $n$\text{-}$\exp^\mathbb{S}_\xi$. The main property of name-expansions that we will use is the following.
	\begin{lemma}
		\label{lem:uexp}
		For any clean formula $\xi$ and bound $\mu$-variable $x_i \in \text{BV}(\xi)$:
		\[
		\text{If $\mathsf{L} \vdash n\text{-}\exp(\delta_{x_i}) \rightarrow \psi_{U_i}$, then $\mathsf{L} \vdash n\text{-}\exp(\mu x_i . \delta_{x_i}) \rightarrow \psi_{U_i}$}.
		\]
	\end{lemma}
	\begin{proof}
		Let $\chi$ be the formula $n\text{-}\exp_\xi(\delta_{x_i})$, but without the substitution $[\psi_{U_i}/x_i]$. Then the to-be-proven implication becomes:
		\[
		\text{If $\mathsf{L} \vdash \chi[\psi_{U_i}/x_i] \rightarrow \psi_{U_i}$, then $\mathsf{L} \vdash \mu x_i . \chi \rightarrow \psi_{U_i}$},
		\]
		but this is simply an application of the least prefixpoint rule.
	\end{proof}
	\begin{lemma}
		\label{lem:innerind}
		Let $\xi$ be a clean formula in $\Sigma$ such that for every free strict subformula of $\xi$ the implication (T) holds. Then for every subformula of $\xi$ of the form $\mu x_i . \delta_{x_i}$ it holds that:
		\[
		\mathsf{L} \vdash n\text{-}\exp(\mu x_i . \delta_{x_i}) \rightarrow \psi_{U_i}.
		\]
	\end{lemma}
	\begin{proof}
		We proceed by induction on the complexity of subformulas of $\xi$. Let $\mu{x_i} .\delta_{x_i}$ be a subformula of $\xi$ and suppose, as inductive hypothesis, that the thesis holds for every strict subformula of $\mu{x_i}.\delta_{x_i}$ (of the form $\mu {x_j}. \delta_{x_j}$). By Lemma \ref{lem:uexp}, it suffices to show that $\mathsf{L} \vdash n\text{-}\exp(\delta_{x_i}) \rightarrow \psi_{U_{i}}$. For this, in turn, it is enough to show that:
		\begin{equation}
		\label{eq:toshowind}
		\text{ For any $A \in S^\Sigma$ such that $\psi_A \land n\text{-}\exp{(\delta_{x_i})}$ is consistent, it holds that $A \Vdash \exp(\delta_{x_i})$.}
		\end{equation} This is because if $\mathsf{L} \vdash n\text{-}\exp(\delta_{x_i}) \rightarrow \psi_{U_i}$ were not the case, then $n\text{-}\exp(\delta_{x_i}) \land \neg \psi_{U_i}$ would be consistent. It follows that there is a maximally $\mathsf{L}$-consistent set $\Gamma$ extending this formula. Letting $A := \Gamma \cap \Sigma$, we obtain that $\psi_A \land n\text{-}\exp(\delta_{x_i})$ is consistent, but $A \not \Vdash \exp(\delta_{x_i})$, for $A \land \psi_{U_i}$ is inconsistent. 
	
		Applying Theorem \ref{thm:gamalg}, we will show (\ref{eq:toshowind}) by constructing a winning strategy for $\exists$ in $\mathcal{G} := \mathcal{E}(\xi, \mathbb{S})$ initialised at $(\delta_{x_i}, A)$. The idea is to show that $\exists$ has a strategy $f$ ensuring for some initial segment of the match that at each position $(\theta, B)$ reached, it holds that $\theta \trianglelefteq \delta_{x_i}$ and the conjunction $\psi_B \land n\text{-}\exp(\theta)$ is consistent. For the rest of this proof we shall call such a position \textit{good}. We then show that this sequence of good positions eventually leads to a good position $(\theta, B)$ such that one of the following holds:
		\begin{enumerate}[label = (\roman*), itemsep = 1pt, topsep = 1pt]
			\item $\theta$ is a free subformula of $\xi$. 
			\item $\theta$ is of the form $\mu x_j. \delta_{x_j}$.
			\item $\theta = x_j$ for some bound variable $x_j$ of $\xi$.
		\end{enumerate}
		Good positions of this form will be called \textit{perfect}. The proof rests on the following three claims:
		\begin{quote}
			\emph{Claim 1.} If $(\theta, B)$ is the last position of some finite $\mathcal{G}$-match $\gamma$ of which every position is good, but not perfect, then $\exists$ can ensure that the next position will also be good. 
		\end{quote}
		\begin{quote}
			\emph{Claim 2.} There can be no infinite match of which every position is good, but not perfect.
		\end{quote}
		\begin{quote}
			\emph{Claim 3.} Any perfect position is winning for $\exists$ in $\mathcal{G}$.
		\end{quote}
		Suppose we have established these three claims. Then, since the initial position of $\mathcal{G}$ is good by assumption, it follows from Claim 1 that $\exists$ can maintain this property until a perfect position is reached. By Claim 2, such a position must be reached after finitely many steps, from which, by Claim 3, there must be some strategy that $\exists$ can take on in order to win the match. \vspace{\baselineskip}
		
		\noindent \emph{Proof of Claim 1.} By a case distinction on the shape of $\theta$, we will show that $\exists$ can ensure that the next position will also be good. First note that $\theta$ cannot have a main connective in $\{\lb, \nu\}$, because then, by items (2) and (3) Lemma \ref{lem:conplay}, the match $\gamma$ must have passed through some formula $\alpha \lhd_f \xi$.  This is impossible since every position in $\gamma$ is assumed not to be perfect. Moreover, the formula $\theta$ can neither be a bound nor a free variable of $\xi$, for in both cases the position $(\theta, B)$ would be perfect. Finally, by the same reason it cannot be the case that $\theta$ is of the form $\mu x_j. \delta_{x_j}$. This leaves the following three cases:
		\begin{itemize}
			\item $\theta$ is of the form $\theta_1 \lor \theta_2$. Then $\psi_B \land n\text{-exp}(\theta_1 \lor \theta_2) = \psi_B \land (n\text{-}\exp(\theta_1)  \lor n\text{-}\exp(\theta_2))$ is consistent, so for some $k \in \{1,2\}$ it must hold that $\psi_B \land n\text{-}\exp(\theta_k)$ is consistent. We let $\exists$ choose accordingly. 
			
			\item $\theta$ is of the form $\theta_1 \land \theta_2$. Then $\psi_B \land n\text{-}\exp(\theta_1 \land \theta_2)$ is consistent. It follows that both the formulas $\psi_B \land n\text{-}\exp(\theta_1)$ and $\psi_B \land n\text{-}\exp(\theta_2)$ are consistent. Thus both moves available to $\forall$ result in good positions. 
			
			\item $\theta$ is of the form $\ld \delta$. Then $n\text{-}\exp(\theta) = \ld n\text{-}\exp(\delta)$, so by the existence lemma there is some $B R C$ such that $(C, n\text{-}\exp(\delta))$ is good, which we let $\exists$ choose. 
			
		\end{itemize}
		\emph{Proof of Claim 2.} This follows from the fact that any infinite $\mathcal{G}$-match must pass through some bound variable of $\xi$. \vspace{\baselineskip}
		
		\noindent \emph{Proof of Claim 3.} Let $(\theta, B)$ be a perfect position in $\mathcal{G}$. We consider the three different types of perfect positions one-by-one. 
		\begin{itemize}
			\item[(i)] In this case we have $\theta \lhd_f \xi$, which means that $\theta \in Cl(\xi) \subseteq \Sigma$. Moreover, since $(\theta, B)$ is good, it holds that $\psi_B \land \theta$ is consistent, whence\ $\theta \in B$. The lemma's hypothesis gives gives $B \Vdash \theta$, supplying $\exists$ with the required strategy. 
			
			\item[(ii)] In this case $\theta$ is of the form $\mu x_j. \delta_{x_j}$. By the fact that $(\theta, B)$ is good, we have that the formula $\psi_B \land n\text{-}\exp(\mu x_j . \delta_{x_j})$ is consistent and $\theta \unlhd \delta_{x_i}$, hence $\theta \lhd \mu x_i . \delta_{x_i}$. Therefore, we can apply the induction hypothesis to conclude that $\psi_B \land \psi_{U_j}$ is consistent. It follows that $B \in [\![\exp(\delta_{x_j})]\!]$, so an application of Theorem \ref{thm:gamalg} gives the required strategy for $\exists$.
			
			\item[(iii)] $\theta$ is a bound variable $x_j$ of $\xi$. Then the fact that $\psi_B \land n\text{-}\exp(x_j)$ is consistent implies that $B \in U_{x_j}$, from which we can obtain the required strategy for $\exists$ in the same way as in the previous case.
		\end{itemize}
	This finishes the proof of Lemma~\ref{lem:innerind}
	\end{proof}
	\noindent \emph{Proof of Lemma \ref{lem:truth}.} We proceed by induction on $\xi$. Suppose that the thesis holds for all subformulas of $\xi$. We will show that $\exists$ has a winning strategy in the game $\mathcal{E}(\xi, \mathbb{S})$ initialised at $(\xi, A)$.
	
	The point is that $\exists$ can initially ensure that at each position $(\theta, B)$ reached, the formula $\psi_B \land \exp_\xi(\theta)$ is consistent (note that by hypothesis this is the case for the initial position). Let $\gamma$ be a match where $\exists$ employs this strategy. If at some point in $\gamma$ a $\mu$-formula is reached, let $(\mu x_i . \delta_{x_i}, B)$ be the first such position. The syntactic restrictions on $\mu_c \mathsf{ML}$ ensure that $\mu x_i . \delta_{x_i}$ will be a free subformula of $\xi$, whence $n\text{-}\exp(\mu x_i . \delta_{x_i}) = \mu x_i . \delta_{x_i}$. Therefore we can invoke Lemma \ref{lem:innerind} to obtain $A \in [\![\exp(\mu x_i . \delta_{x_i})]\!]$. Theorem \ref{thm:gamalg} supplies $\exists$ with a strategy to follow from here on out. 
		
	Now suppose that no $\mu$-formula is reached in some complete match $\gamma$. If $\gamma$ is infinite, it must be winning for $\exists$. Finally, if $\gamma$ is finite, the player $\forall$ must have gotten stuck, or at some point a free variable of $\xi$ is reached. The latter is also winning for $\exists$ because of the assumption that for every position $(\theta, B)$ reached, it holds that $\psi_B \land \exp(\theta)$ is consistent. \hfill \qed
	\paragraph{Completeness}
	The goal of this paragraph is to prove completeness for certain well-behaved $\mu_c$-logics.
	
	Given a logic $\mathsf{L}$, we define its canonical model as usual. 
	\begin{definition}
		The \textit{canonical model} $\mathbb{S}^\mathsf{L} := (S^{\mathsf L}, R^{\mathsf L}, V^{\mathsf L})$ of a logic $\mathsf{L}$ is given by:
		\begin{multicols}{2}	
		\begin{itemize}
			\item $S^\mathsf{L} := \{\Gamma : \Gamma \text{ is maximally $\mathsf{L}$-consistent}\}$.
			\item $\Gamma R^\mathsf{L} \Delta :\Leftrightarrow (\lb \varphi \in \Gamma \Rightarrow \varphi \in \Delta)$.
			\item $V^\mathsf{L}(p) := \{\Gamma : p \in \Gamma\}$.
			\item[]  \hfill \qedhere
		\end{itemize}
	\end{multicols}
	\end{definition}
	For (infinitary) canonical models there is also a standard existence lemma:
	\begin{lemma}
	\label{lem:existencecan}
		For any state $\Gamma$ of a canonical model $\mathbb{S}^\mathsf{L}$: 
		\[
		\textnormal{If $\ld \varphi \in \Gamma$, then there is a state $\Delta$ such that $\Gamma R^\mathsf{L} \Delta$ and $\varphi \in \Delta$.}
		\]
	\end{lemma}
	Generally, a $\mu_c$-logic $\mathsf{L}$ will lack the compactness property. It is well-known that this prevents one to prove a Truth Lemma for the (standard) canonical model of $\mathsf{L}$. Indeed, if there are unsatisfiable sets of formulas which are finitely satisfiable, then, because derivations are finite objects, there will be maximally consistent sets which are unsatisfiable. 
	
	The following lemma is analogous to Lemma \ref{lem:admitbasicadmitmucforlogics}.
		\begin{lemma}
		\label{lem:iso}
		Let $\mathsf{L}$ be a logic and let $\mathcal{F}$ be a class of frames that admits filtration and contains the canonical frame $(S^\mathsf{L}, R^\mathsf{L})$. For any finite and $\sim$FL-closed set $\Sigma$, the class $\mathcal{F}$ contains a frame underlying some model $\mathbb{S}$ over $\Sigma$ with respect to $\mathsf{L}$.
	\end{lemma}
	\begin{proof}
		We will apply a form of filtration to the canonical model $\mathbb{S}^\mathsf{L}$. As in the proof of Lemma \ref{lem:admitbasicadmitmucforlogics}, we let $\varphi_1, \ldots, \varphi_n$ be an enumeration of the formulas of the form $\eta x . \psi$ in $\Sigma$. For every such formula $\varphi_i$, we pick a unique propositional variable $p_i$ not occurring in $\Sigma$.
		
		We will define an alternative valuation function $V' : \mathsf{P} \rightarrow \mathcal{P}(S^\mathsf{L})$. In contrast to the proof of Lemma \ref{lem:admitbasicadmitmucforlogics}, we will not let the valuation of $p_i$ be the \textit{meaning} of $\varphi_i$ in $\mathbb{S}^\mathsf{L}$, but rather we let $p_i$ be true at those $\Gamma$ for which $\varphi_i \in \Gamma$. Note that if a Truth Lemma would hold for $\mathbb{S}^\mathsf{L}$, these two options would be equivalent.
		\[
		V'(p) := \begin{cases} \{\Gamma : \varphi_i \in \Gamma\} &\text{ if $p = p_{i}$ for some $\varphi_i \in \Sigma$;} \\
			V(p) &\text{otherwise.}
		\end{cases}
		\]
		Let $\mathbb{S}'$ be the model $\mathbb{S}^\mathsf{L}$, but with $V'$ as valuation function.	We define the translation $\tau : \Sigma \rightarrow \mathsf{\mathsf{ML}}$ in the same way as we did in the proof of Lemma \ref{lem:admitbasicadmitmucforlogics}. The set $\tau[\Sigma]$ is again $\sim$FL-closed. A straightforward induction shows that for every $\xi \in \Sigma$:
		\begin{equation}
			\label{eq:subs}
			\mathbb{S}', \Gamma \Vdash \xi \Leftrightarrow \tau(\xi) \in \Gamma.
		\end{equation}
		Since the frame underlying $\mathbb{S}'$ is in $\mathcal{F}$, we can apply the assumed admissibility of filtration to obtain a filtration $\mathbb{S}^{\tau[\Sigma]} = (S^{\tau[\Sigma]}, R^{\tau[\Sigma]}, V^{\tau[\Sigma]})$ of $\mathbb{S}'$ through ${\tau[\Sigma]}$ such that the frame $(S^{\tau[\Sigma]}, R^{\tau[\Sigma]})$ belongs to $\mathcal{F}$. 
		
		We will finish the proof by showing that $\mathbb{S}^{\tau[\Sigma]}$ is isomorphic to a model over $\Sigma$. We define the set of states $S^\Sigma := \{\Gamma \cap \Sigma : \Gamma \text{ is maximally $\mathsf{L}$-consistent}\}$ and claim that the map
		\[
		h : [\Gamma] \mapsto \Gamma \cap \Sigma
		\]
		is a well-defined bijection from $S^\mathsf{L}/ \sim^{\mathbb{S}'}_{\tau[\Sigma]}$ to $S^\Sigma$.	For well-definedness, suppose $\Gamma \sim_{\tau[\Sigma]}^{\mathbb{S}'} \Gamma'$ and let $\varphi \in \Sigma$. Using the equivalence (\ref{eq:subs}), we have
		\[
		\varphi \in \Gamma \Leftrightarrow \mathbb{S}', \Gamma \Vdash \tau(\varphi) \Leftrightarrow \mathbb{S}', \Gamma' \Vdash \tau(\varphi) \Leftrightarrow \varphi \in \Gamma',
		\]
		as required. 
		
		Injectivity is similar: if $\Gamma \cap \Sigma = \Gamma' \cap \Sigma$, then for all $\tau(\varphi) \in \tau[\Sigma]$, we have: 
		\[
		\Gamma \Vdash \tau(\varphi) \Leftrightarrow \varphi \in \Gamma \Leftrightarrow \varphi \in \Gamma' \Leftrightarrow \Gamma' \Vdash \tau(\varphi) . 
		\]
		For surjectivity, take $\Gamma \cap \Sigma$ for some any $\Gamma \in S^\mathsf{L}$. Then $h([\Gamma]) = \Gamma \cap \Sigma$, as required. 
		
		Now let the relation $R^\Sigma \subseteq S^\Sigma \times S^\Sigma$ and the valuation $V^\Sigma : \mathsf{P} \rightarrow \mathcal{P}(S^\Sigma)$ be given by transporting the structure of $\mathbb{S}^{\tau[\Sigma]}$ along $h$. More precise, we let 
		\[
		AR^\Sigma B :\Leftrightarrow h^{-1}(A) R^{\tau[\Sigma]} h^{-1}(B).
		\]
		We claim that $R^\mathsf{min} \subseteq R^\Sigma \subseteq R^\mathsf{max}$. 
		
		First, suppose that $A R^\mathsf{min} B$. Then $\psi_A \land \ld \psi_B$ is $\mathsf{L}$-consistent. Pick some $\Gamma \in S^\mathsf{L}$ containing both $\psi_A$ and $\ld \psi_B$. By Lemma \ref{lem:existencecan}, there is a $\Delta \in S^\mathsf{L}$ such that $\Gamma R^\mathsf{L} \Delta$ and $\psi_B \in \Delta$. Since $R^{\tau[\Sigma]}$ contains the finest filtration, we have $[\Gamma] R^{\tau[\Sigma]} [\Delta]$ and thus $h([\Gamma])R^\Sigma h([\Delta])$. The required result follows from the fact that $h([\Gamma]) = A$ and $h([\Delta]) = B$. 
		
		Now suppose that $AR^\Sigma B$. We will show that $A R^\mathsf{max} B$. To that end, let $\lb \varphi \in \Sigma$ such that $\lb \varphi \in A$. Pick $\Gamma \supset A$ and $\Delta \supset B$ from $S^\mathsf{L}$. Since $[\Gamma] = h^{-1}(A)$ and $[\Delta] = h^{-1}(B)$, we have $[\Gamma]R^{\tau[\Sigma]}[\Delta]$. We now use the fact that $R^{\tau[\Sigma]}$ is contained in the coarsest filtration. This means that for all $\lb \psi \in \tau[\Sigma]$ such that $\mathbb{S}', \Gamma \Vdash \lb \psi$, we have $\mathbb{S}', \Delta \Vdash \psi$. By assumption we have $\lb \varphi \in \Gamma$, whence the equivalence (\ref{eq:subs}) gives $\mathbb{S}', \Gamma \Vdash \tau(\lb \varphi)$, \textit{i.e.} $\mathbb{S}', \Gamma \Vdash \lb\tau (\varphi)$. It follows that $\mathbb{S}', \Delta \Vdash \tau(\varphi)$. Finally, another application of the equivalence (\ref{eq:subs}) yields $\varphi \in \Delta$, hence $\varphi \in B$, as required.
		
		Lastly, for any $p \in \Sigma$, we define 
		\[
		V^\Sigma(p) := \{A \in S^\Sigma : h^{-1}(A) \in V^{\tau[\Sigma]}(p)\} = \{A \in S^\Sigma : p \in A\}, 
		\]
		which suffices. 
	\end{proof}
	\begin{theorem}
		Let $\mathsf{L}$ be a canonical logic in the basic modal language such $\mathsf{Fr}(\mathsf{L})$ admits filtration. Then $\mu_c$-$\mathsf{L}$ is sound and complete with respect to $\mathsf{Fr}(\mathsf{L})$.
	\end{theorem}	
	\begin{proof}
		Soundness follows from the fact the fixpoint axioms and rules are sound on the class of all frames. For completeness, let $\varphi \in \mu_c \mathsf{ML}$ be $\mathsf{L}$-consistent; we will show that $\varphi$ is satisfiable in a model based on a $\mathsf{L}$-frame. Without loss of generality we may assume that $\varphi$ is clean. Let $\Sigma$ be the $\sim$FL-closure of $\{\varphi\}$. Note that by canonicity the canonical frame $(S^\mathsf{L}, R^\mathsf{L})$ is contained in $\mathsf{Fr}(\mathsf{L})$. Therefore, we can use Lemma \ref{lem:iso} to obtain a model $\mathbb{S}^\Sigma$ over $\Sigma$ with respect to $\mathsf{L}$ which is based on an $\mathsf{L}$-frame. By the $\mathsf{L}$-consistency of $\varphi$, there is a state $A \in S^\Sigma$ such that $\varphi \in A$. Finally, Lemma \ref{lem:truth} gives $\mathbb{S}^\Sigma, A \Vdash \varphi$, as required.
	\end{proof}
	For instance, the logic $\mu_c$-$\mathsf{KB}$ is sound and complete with respect to the class of symmetric frames. Some other examples of basic modal logics that satisfy the hypotheses of the above theorem are: $\mathsf{K}$, $\mathsf{T}$, $\mathsf{K4}$, $\mathsf{S4}$ and $\mathsf{S5}$.
	
	\bibliographystyle{eptcs}
	\bibliography{generic}

\begin{thebibliography}{10}
\providecommand{\bibitemdeclare}[2]{}
\providecommand{\surnamestart}{}
\providecommand{\surnameend}{}
\providecommand{\urlprefix}{Available at }
\providecommand{\url}[1]{\texttt{#1}}
\providecommand{\href}[2]{\texttt{#2}}
\providecommand{\urlalt}[2]{\href{#1}{#2}}
\providecommand{\doi}[1]{doi:\urlalt{http://dx.doi.org/#1}{#1}}
\providecommand{\bibinfo}[2]{#2}

\bibitemdeclare{article}{bent:moda06}
\bibitem{bent:moda06}
\bibinfo{author}{{Johan} \surnamestart van Benthem\surnameend}
  (\bibinfo{year}{2006}): \emph{\bibinfo{title}{Modal frame correspondences and
  fixed-points}}.
\newblock {\sl \bibinfo{journal}{Studia Logica}} \bibinfo{volume}{83}, pp.
  \bibinfo{pages}{133--155}, \doi{10.1007/s11225-006-8301-9}.

\bibitemdeclare{misc}{vanmodern}
\bibitem{vanmodern}
\bibinfo{author}{Johan \surnamestart van Benthem\surnameend} \&
  \bibinfo{author}{Nick \surnamestart Bezhanishvili\surnameend}:
  \emph{\bibinfo{title}{Modern faces of filtration}}.
\newblock \bibinfo{note}{ILLC Prepublication PP-2019-13}.

\bibitemdeclare{phdthesis}{carreiro2015fragments}
\bibitem{carreiro2015fragments}
\bibinfo{author}{Facundo \surnamestart Carreiro\surnameend}
  (\bibinfo{year}{2015}): \emph{\bibinfo{title}{Fragments of fixpoint logics}}.
\newblock Ph.D. thesis, \bibinfo{school}{University of Amsterdam}.

\bibitemdeclare{article}{carr:powe20}
\bibitem{carr:powe20}
\bibinfo{author}{Facundo \surnamestart Carreiro\surnameend},
  \bibinfo{author}{Alessandro \surnamestart Facchini\surnameend},
  \bibinfo{author}{Yde \surnamestart Venema\surnameend} \&
  \bibinfo{author}{Fabio \surnamestart Zanasi\surnameend}
  (\bibinfo{year}{2020}): \emph{\bibinfo{title}{The Power of the Weak}}.
\newblock {\sl \bibinfo{journal}{{ACM} Transactions on Computational Logic}}
  \bibinfo{volume}{21}(\bibinfo{number}{2}), pp. \bibinfo{pages}{15:1--15:47},
  \doi{10.1016/S0304-3975(01)00185-2}.

\bibitemdeclare{inproceedings}{fontaine2008continuous}
\bibitem{fontaine2008continuous}
\bibinfo{author}{Ga{\"e}lle \surnamestart Fontaine\surnameend}
  (\bibinfo{year}{2008}): \emph{\bibinfo{title}{Continuous fragment of the
  mu-calculus}}.
\newblock In: {\sl \bibinfo{booktitle}{International Workshop on Computer
  Science Logic}}, \bibinfo{organization}{Springer}, pp.
  \bibinfo{pages}{139--153}, \doi{10.1007/3-540-49116-3_50}.

\bibitemdeclare{book}{goldblatt1987logics}
\bibitem{goldblatt1987logics}
\bibinfo{author}{Robert \surnamestart Goldblatt\surnameend}
  (\bibinfo{year}{1987}): \emph{\bibinfo{title}{Logics of time and
  computation}}.
\newblock \bibinfo{publisher}{Center for the Study of Language and
  Information}.

\bibitemdeclare{inproceedings}{jani:expr96}
\bibitem{jani:expr96}
\bibinfo{author}{David \surnamestart Janin\surnameend} \& \bibinfo{author}{Igor
  \surnamestart Walukiewicz\surnameend} (\bibinfo{year}{1996}):
  \emph{\bibinfo{title}{On the Expressive Completeness of the Propositional
  $\mu$-Calculus w.r.t.\ Monadic Second-Order Logic}}.
\newblock In: {\sl \bibinfo{booktitle}{Proceedings~of the Seventh International
  Conference on Concurrency Theory, CONCUR '96}}, {\sl \bibinfo{series}{LNCS}}
  \bibinfo{volume}{1119}, pp. \bibinfo{pages}{263--277},
  \doi{10.1007/3-540-61604-7_60}.

\bibitemdeclare{inproceedings}{kikot2020completeness}
\bibitem{kikot2020completeness}
\bibinfo{author}{Stanislav \surnamestart Kikot\surnameend},
  \bibinfo{author}{Ilya \surnamestart Shapirovsky\surnameend} \&
  \bibinfo{author}{Evgeny \surnamestart Zolin\surnameend}
  (\bibinfo{year}{2020}): \emph{\bibinfo{title}{Modal Logics with Transitive
  Closure: Completeness, Decidability, Filtration}}.
\newblock In \bibinfo{editor}{Nicola \surnamestart Olivetti\surnameend},
  \bibinfo{editor}{Rineke \surnamestart Verbrugge\surnameend},
  \bibinfo{editor}{Sara \surnamestart Negri\surnameend} \&
  \bibinfo{editor}{Gabriel \surnamestart Sandu\surnameend}, editors: {\sl
  \bibinfo{booktitle}{13th Conference on Advances in Modal Logic, AiML 2020,
  Helsinki, Finland, August 24-28, 2020}}, \bibinfo{publisher}{College
  Publications}, pp. \bibinfo{pages}{369--388}.

\bibitemdeclare{article}{kozen1983results}
\bibitem{kozen1983results}
\bibinfo{author}{Dexter \surnamestart Kozen\surnameend} (\bibinfo{year}{1983}):
  \emph{\bibinfo{title}{Results on the propositional $\mu$-calculus}}.
\newblock {\sl \bibinfo{journal}{Theoretical computer science}}
  \bibinfo{volume}{27}(\bibinfo{number}{3}), pp. \bibinfo{pages}{333--354},
  \doi{10.1016/0304-3975(82)90125-6}.

\bibitemdeclare{article}{kozen1981elementary}
\bibitem{kozen1981elementary}
\bibinfo{author}{Dexter \surnamestart Kozen\surnameend} \&
  \bibinfo{author}{Rohit \surnamestart Parikh\surnameend}
  (\bibinfo{year}{1981}): \emph{\bibinfo{title}{An elementary proof of the
  completeness of PDL}}.
\newblock {\sl \bibinfo{journal}{Theoretical Computer Science}}
  \bibinfo{volume}{14}(\bibinfo{number}{1}), pp. \bibinfo{pages}{113--118},
  \doi{10.1016/0304-3975(81)90019-0}.

\bibitemdeclare{book}{lemmon1977introduction}
\bibitem{lemmon1977introduction}
\bibinfo{author}{John \surnamestart Lemmon\surnameend} \& \bibinfo{author}{Dana
  \surnamestart Scott\surnameend} (\bibinfo{year}{1977}):
  \emph{\bibinfo{title}{An introduction to modal logic}}.
\newblock \bibinfo{publisher}{Blackwell}.

\bibitemdeclare{misc}{martivenemafocus}
\bibitem{martivenemafocus}
\bibinfo{author}{Johannes \surnamestart Marti\surnameend} \&
  \bibinfo{author}{Yde \surnamestart Venema\surnameend} (\bibinfo{year}{2021}):
  \emph{\bibinfo{title}{Focus-style proof systems and interpolation for the
  alternation-free $\mu$-calculus}}.

\end{thebibliography}
\end{document}